\documentclass[preprint]{elsarticle}

\usepackage{amsmath}
\usepackage{amssymb}

\usepackage{hyperref}

\let\emptyset=\varnothing

\newcommand{\ie}{\textrm{i.e., }}

\newtheorem{theorem}{Theorem}
\newtheorem{corollary}{Corollary}

\newtheorem{definition}[theorem]{Definition}
\newdefinition{example}[theorem]{Example}

\newenvironment{proof}{\emph{Proof.}}{\qed \vspace{2mm}}

\begin{document}


\title{A Second-Order Formulation of Non-Termination}

\author{Fred~Mesnard}
\ead{frederic.mesnard@univ-reunion.fr}

\author{\'Etienne~Payet}
\ead{etienne.payet@univ-reunion.fr}

\address{Universit\'e de La R\'eunion, EA2525-LIM,
  Saint-Denis de La R\'eunion, F-97490, France}

\begin{abstract}
We consider the termination/non-termination property of a class of loops.
Such loops are commonly used abstractions of real program pieces.
Second-order logic is a convenient language to express non-termination. 
Of course, such property is generally undecidable. 
However, by restricting the language to known decidable cases,
we exhibit new classes of loops, the non-termination of which is decidable. 
We present a bunch of examples.
\end{abstract}


\maketitle


\section{Introduction}

%

In this paper, we recall that second-order logic is a convenient
language to express non-termination of while loops, modeled as
rules.  Such rules are commonly used abstractions of real program pieces,
see, e.g., \cite{SpotoMP10} for the Java programming language. 
Our main contribution is the definition of two new classes of rules,
the termination of which is decidable, by restricting the language
to known decidable cases, namely S1S and S2S.
We also show and illustrate how decision procedures for their weak
versions WS1S and WS2S can help proving termination/non-termination.

We organize the paper as follows.
Section \ref{sec:preliminaries} presents the main concepts we need while
Section \ref{sec:sof-nt} gives the theoretical results of the paper. 
Section \ref{sec:examples} illustrates the results by means of examples and 
Section \ref{sec:related-works-conclusion} concludes.

\section{Preliminaries}
\label{sec:preliminaries}
We give a quick description of S1S and S2S, see~\cite{Thomas90} for a more
detailed presentation.
S1S is the monadic Second-order theory of 1 Successor.
Interpretations correspond to finite or infinite words
over a given finite alphabet $\Sigma$.
Terms are constructed from the constant 0 and first-order
variables $x$, $y$, \ldots{} by application of the successor
function $+1$, which is left-associative. We abbreviate $n$
successive applications of $+1$ starting from $0$
(\ie $0+1+1+\cdots+1$) to $n$.
Atomic formul\ae{} are constructed from terms,
second-order variables $X$, $Y$, \ldots{}
and predicates of the form $P_a$ where $a\in\Sigma$.
They have the form $t=t'$, $t<t'$, $t\in X$, $P_a(t)$
where $t$ and $t'$ are terms.
Formul\ae{} are constructed from atomic formul\ae{},
the usual boolean connectives ($\lor$, $\land$, \ldots)
and quantification ($\forall$ and $\exists$) over first
and second-order variables.
First-order variables are interpreted as elements of
$\mathbb{N}$ representing positions in words and second-order
variables as subsets of $\mathbb{N}$. Constant $0$
is interpreted as the first position in a word and function
$+1$ as the next position.
The formula $P_a(t)$ is true in a word $w$ if at position
$t$ of $w$ there is character $a$.
WS1S (Weak S1S) is a restriction of S1S where
second-order variables are interpreted as \emph{finite}
sets only.

S2S is the monadic Second-order theory of 2 Successors.
Interpretations correspond to finite or infinite labelled
binary trees over a given finite alphabet $\Sigma$.
Terms and formul\ae{} are constructed as in S1S
except that constant 0 is replaced with $\varepsilon$
and the successor function $+1$ is replaced with functions
$.0$ and $.1$, which are left-associative. We abbreviate
successive applications of these functions,
for instance $x.0110$ stands for $x.0.1.1.0$, which
corresponds to $(((x.0).1).1).0$, and $0110$ stands
for $\varepsilon.0.1.1.0$.
First-order variables are interpreted as elements of
$\{0,1\}^*$ representing positions in binary trees and second-order
variables as subsets of $\{0,1\}^*$.
Constant $\varepsilon$ is interpreted as the root position
of a binary tree, $.0$ as the left successor, $.1$ as the
right successor and $<$ as the proper-prefix relation
(for instance $01 < 0110$ but $00 \not< 0110$).
WS2S (Weak S2S) is a restriction of S2S where
second-order variables are interpreted as \emph{finite}
sets only.

A \emph{rule} has the form 
$r:~\tilde{x} \rightarrow \psi(\tilde{x},\tilde{y}), \tilde{y}$
where $r$ is the identifier of the rule, 
$\psi$ is a binary relation and $\tilde{x}$ and $\tilde{y}$
are tuples of distinct first-order variables ranging over a given domain.
If  $r$ is a monadic rule of the form $x \rightarrow \psi(x,y), y$
and $\psi(x,y)$ is a monadic second-order formula of S1S (S2S)
with $x$ and $y$ as free variables, we call $r$  a
\emph{monadic} S1S (respectively, S2S) rule. 
Some examples can be found in Section~\ref{sec:examples}.
We define an operational semantics as follows. 
Starting from a concrete tuple $\tilde{x_0}$ of elements of the domain,
we first check whether there exists a concrete tuple $\tilde{x_1}$ such that
$\psi(\tilde{x_0}, \tilde{x_1})$.  If no such tuple exists, the computation stops.
Otherwise, we choose any such tuple $\tilde{x_1}$ and reiterate.
The rule $r$ \emph{loops} if we can find a concrete tuple $\tilde{x_0}$ starting an infinite
computation. If no such tuple exists, $r$ \emph{terminates}.

\section{A Second-Order Formulation of Non-Termination}
\label{sec:sof-nt}
We consider the following second-order formulation of non-termination. Let 
$r:~\tilde{x} \rightarrow \psi(\tilde{x},\tilde{y}), \tilde{y}$ be a rule. 

\begin{definition}[recurrence set \cite{Gupta08}]
  \label{def:recurrence-set}
  We let $\phi_r$ denote the second-order 
  formula
  \[\exists X \left\{ \begin{array}{ll}
      \exists \tilde{x}\ \tilde{x}\in X \land \phantom{x} & (1) \\[1ex]
      \forall \tilde{x} \exists \tilde{y}\ (\tilde{x}\in X
      \Rightarrow [\psi(\tilde{x},\tilde{y}) \land \tilde{y}\in X])\; & (2) 
    \end{array}\right.\]
  A \emph{recurrence set} for $r$ is a set $X$ satisfying $\phi_r$.
\end{definition}
Condition (1) of  Definition~\ref{def:recurrence-set} simply states that the
recurrence set $X$ is not empty. Condition (2) ensures that for any element
$x$ of $X$, there is an element $y$ of $X$ which satisfies the formula
$\psi(\tilde{x},\tilde{y})$ defining the rule $r$.
The existence of a recurrence set is equivalent to non-termination.

\begin{theorem}[\cite{Gupta08}]
  \label{theorem:nonterm-cns}
  $\phi_r$ is true if and only if $r$  loops.
\end{theorem}
\begin{proof} We prove both implications.

\noindent ($\Rightarrow$). As $\phi_r$ is true, we can start by
selecting any arbitrary $\tilde{x_0} \in X$.
We know that there exists $\tilde{y_0} \in X$ s.t. $\psi(\tilde{x_0},\tilde{y_0})$.
By iterating this process, we construct an infinite computation.
Hence $r$ loops.

\noindent ($\Leftarrow$). As  there exists $\tilde{x_0}$ such that $r$ loops, 
let us consider an infinite computation starting at $\tilde{x_0}$: 
$\tilde{x_0}, \tilde{x_1}, \ldots, \tilde{x_n}, \ldots$
Let $X = \{ \tilde{x_i} | i \ge 0 \}$. $X$ is a non-empty set verifying 
$\forall \tilde{x} \exists \tilde{y}\ (\tilde{x}\in X \Rightarrow
[\psi(\tilde{x},\tilde{y}) \land \tilde{y}\in X])$.
Hence $\phi_r$ holds.
\end{proof}

The notion of \emph{closed recurrence set} is introduced
in~\cite{ChenCFNO14}.
\begin{definition}[closed recurrence set]
  \label{def:closed-recurrence-set}
  We let $\phi'_r$ be the second-order 
  formula
  \[\exists X \left\{ \begin{array}{ll}
      \exists \tilde{x}\ \tilde{x}\in X \land \phantom{x} & (1) \\[1ex]
      \forall \tilde{x} \exists \tilde{y}\ (\tilde{x}\in X
      \Rightarrow \psi(\tilde{x},\tilde{y})) \land \phantom{x} & (2) \\[1ex]
     \forall \tilde{x}\forall\tilde{y}\
      ([\tilde{x}\in X \land \psi(\tilde{x},\tilde{y})]
      \Rightarrow \tilde{y}\in X) & (3)
    \end{array}\right.\]
  A \emph{closed recurrence set} for $r$ 
  is a set $X$ satisfying $\phi'_r$.
\end{definition}
Let $S_{\tilde{x}}$ denote the set of successors of $\tilde{x}$ with
$\psi$, \ie the set of $\tilde{y}$ such that $\psi(\tilde{x},\tilde{y})$
holds. Definition~\ref{def:closed-recurrence-set} imposes that
$\emptyset \neq S_{\tilde{x}} \subseteq X$ for any $\tilde{x}$ in $X$.
Hence, we \emph{must} stay in $X$ when applying $\psi$ to $\tilde{x}$.
In contrast, Definition~\ref{def:recurrence-set} requires that
$S_{\tilde{x}} \cap X \neq\emptyset$, \ie we \emph{can} stay in $X$
when applying $\psi$ to $\tilde{x}$. Therefore, closed recurrence sets
are recurrence sets but we also have that recurrence sets always contain
closed recurrence sets.
\begin{theorem}[\cite{ChenCFNO14}]
  \label{theorem:open-iff-closed}
  If there is a recurrence set $X$ for $r$
  then there exists a rule
  $r':~\tilde{x} \rightarrow \psi'(\tilde{x},\tilde{y}), \tilde{y}$
  with $\psi'\Rightarrow \psi$ and $X'\subseteq X$
  such that $X'$ is a closed recurrence set for $r'$.
\end{theorem}

The second-order formula $\phi_r$ is a necessary and sufficient
condition for non-termination of at least one of the computations $r$ can generate. Symmetrically,
$\neg \phi_r$ is true if and only if for every value $\tilde{x_0}$, any computation
starting at $\tilde{x_0}$ halts.
As such a problem is in general undecidable (see, e.g., \cite{BradleyMS05}),
it follows that $\phi_r$ is not computable. However, when the second-order
logic is restricted to decidable cases,  we obtain classes of rules for
which the termination/non-termination property is decidable.

\begin{theorem}
  Termination of a \emph{monadic} S1S or S2S rule is decidable.
\end{theorem}
\begin{proof}
The monadic second-order logics S1S and S2S are decidable 
\cite{Buchi62S1S, Rabin69S2S}
and so is $\phi_r$ for a monadic S1S or S2S rule $r$. 
If $\phi_r$ is true then $r$ loops else $r$ terminates.
\end{proof}

Weak versions of these logics, where second-order variables range over 
\emph{finite} sets, are also decidable 
and decision procedures have been implemented (see, e.g., MONA \cite{monamanual2001}).
Let $r$ be a monadic S1S or S2S rule.

\begin{corollary}
\label{dec-proc-w-computable-sufficient-conditions}
Decision procedures for WS1S and WS2S provide computable sufficient conditions
for proving non-termination of $r$ in the corresponding structure.
\end{corollary}
\begin{proof}
If such a decision procedure states that $\phi_r$ is true, 
then we know that there exists a non-empty 
finite set $X$ such that $\phi_r$ holds. Hence $r$ loops.
\end{proof}

Note that if the decision procedure states that $\phi_r$ is false, then there is no finite set $X$
satisfying $\phi_r$ but an infinite set $X$ satisfying $\phi_r$ may exist. 
Hence we cannot conclude, except in the following case.

\begin{corollary}
\label{dec-proc-w-dec-term}
When we
know that the set of points which can start a computation  from $r$ is finite,
decision procedures for WS1S and WS2S also decide termination of $r$ 
in the corresponding structure.
\end{corollary}
\begin{proof}
If a decision procedure states that $\phi_r$ is true, 
then by Corollary \ref{dec-proc-w-computable-sufficient-conditions} $r$ loops. 
Else it states that $\phi_r$ is false. So there does not exist a
finite set $X$ satisfying $\phi_r$. As $X$ cannot be infinite by hypothesis, 
it means that there does not exist a set $X$ such that $\phi_r$ holds. Hence $r$ terminates.
\end{proof}

Note that the condition of Corollary \ref{dec-proc-w-dec-term} can be decided in WS1S as it can be stated as
$\exists m \ \forall x \ (x > m \Rightarrow  \neg \ \exists y \ \psi(x,y))$. However Example 
\ref{corollary-does-not-decide-term} shows that it does not decide termination.

\section{Examples}
\label{sec:examples}

\begin{example}[S1S]
  Consider $r:~x \rightarrow \psi(x,y), y$ where
  \[\psi(x,y) =  (3 < x \land x < 10 \land y < x) \lor (x < 3 \land y = x+1) \]
  The set  of points which can start a computation  from $r$ is finite: $\{x \in \mathbb{N} | x \neq 3 \land x < 10\}$.
  MONA tells us that $\phi_r$ is false. By Corollary~\ref{dec-proc-w-dec-term}, $r$ terminates.
  \qed
\end{example}

\begin{example}[S1S]
  Consider $r:~x \rightarrow \psi(x,y), y$ where
  \[\psi(x,y) = (3 < x  \land y < x) \lor (x < 4  \land y = x+1) \]
  MONA reports that $\phi_r$ is true, with a computed satisfying $X=\{3, 4\}$.
  Indeed for any $x \in X=\{3, 4\}$, 
  there is a $y$ in $X$ such that $\psi(x,y)$ holds: if $x=3$, take $y=4$ and if $x=4$, $y=3$.
  Note that the set $X$ is not unique, as $\phi_r$ is also true for, e.g., $X=\{2, 3, 4, 2014\}$.
  By Corollary~\ref{dec-proc-w-computable-sufficient-conditions}, $r$ loops. 
  \qed
\end{example}

\begin{example}[S1S]
\label{corollary-does-not-decide-term}
  Consider $r:~x \rightarrow \psi(x,y), y$ where
  \[\psi(x,y) =  (x <  y)\]
  Although MONA tells us that there is no finite $X$ satisfying $\phi_r$, 
  as the set  of points which can start a computation is infinite, we cannot apply
  Corollary~\ref{dec-proc-w-dec-term}.
  Indeed, taking $X=\mathbb{N}$ shows that $\phi_r$ is true. 
  Hence by Theorem~\ref{theorem:nonterm-cns}, $r$ loops.
  Note that any decision procedure for S1S will prove that $\phi_r$ is true.\qed
\end{example}

\begin{example}[S1S]
  Consider $r:~x \rightarrow \psi(x,y), y$ where
  \[\psi(x,y) =  (y <  x)\]
  MONA reports that there is no finite $X$ satisfying $\phi_r$.
  As the set  of points which can start a computation is infinite, we cannot apply
  Corollary~\ref{dec-proc-w-dec-term}.
  Assume that $\phi_r$ is true. So there is a non-empty $X \subseteq \mathbb{N}$
  satisfying $\phi_r$. Let $e$ be its least element. Condition (2) of Definition~\ref{def:recurrence-set} 
  states that there exists $d$ in $X$ such that $d < e$, which contradicts that $e$ is the least element of $X$.
  Hence $\phi_r$ is false, as should be shown by any decision procedure for S1S. 
  By Theorem~\ref{theorem:nonterm-cns}, $r$ terminates.\qed
\end{example}

\begin{example}[S1S]
  Consider $r:~x \rightarrow \psi(x,y), y$ where
    \[\psi(x,y) =  (\forall X \ (x \in X \land \psi'(X)) \Rightarrow y \in X)\]
  with
  \[\psi'(X) = (\forall z\ z \in X \Rightarrow z + 1 \in X)\]
  We have $\psi'(X)$ is true if and only if $X$ is closed by
  application of the successor function $+1$. So,
  $\psi(x,y)$ is true if and only if $x\leq y$.
  MONA reports that $\phi_{r}$ is true, with a computed satisfying $X=\{0\}$.
  By Corollary~\ref{dec-proc-w-computable-sufficient-conditions}, $r$ loops.
  \qed
\end{example}

\begin{example}[S2S]
  Consider $r:~x \rightarrow \psi(x,y), y$ where
  \[\psi(x,y) = (y = x.1 \lor x = y.1)\]
  The set $X=\{\varepsilon, 1\}$ is not empty and for any $x$
  in $X$ there is a $y$ in $X$ such that $\psi(x,y)$ holds.
  So $\phi_r$ is true (also shown by MONA).
  By Corollary~\ref{dec-proc-w-computable-sufficient-conditions},
  $r$ loops.
  
  Note that $X$ is a recurrence set for $r$ which is
  not closed. Indeed, $1\in X$ and $\psi(1,1^2)$ holds with
  $1^2\not\in X$. Hence, condition (3) of
  Definition~\ref{def:closed-recurrence-set} does not hold.
  By Theorem~\ref{theorem:open-iff-closed}, there must exist 
  $r':~x \rightarrow \psi'(x,y), y$ with
  $\psi'\Rightarrow \psi$ and $X'\subseteq X$ such that $X'$
  is a closed recurrence set for $r'$. For 
  \[\psi'(x,y) = (x\neq 1 \land y = x.1) \lor
  x = y.1  \quad\text{and}\quad X'=X\]
  we have that $\psi'\Rightarrow \psi$,
  $X'\subseteq X$ and $X'$ is a closed recurrence
  set for $r'$.
  \qed
\end{example}

\begin{example}[S2S]
  Consider $r:~x \rightarrow \psi(x,y), y$ where
  \begin{eqnarray*}
    \lefteqn{\psi(x,y) =}\\
    & & (\exists z\ x < 0^4 \land x = z.0 \land y = z.1) \lor \\
    & & (\exists z\  z.01 \leq x \land y = z.11) \lor \\
    & & (x = 1^2 \land y = 0^3)
  \end{eqnarray*}
  The set $X = \{1^2, 0^3, 0^21, 01^2\}$ is not empty and for
  any $x$ in $X$ there is a $y$ in $X$ such that $\psi(x,y)$
  holds. Hence $\phi_r$ is true (also shown by MONA), so $r$ loops. 
  Note that $X$ is a closed recurrence set for $r$. \qed
\end{example}

\begin{example}[S2S]
  Consider $r:~x \rightarrow \psi(x,y), y$ where
  \begin{eqnarray*}
    \lefteqn{\psi(x,y) =}\\
    & & (\exists z\ x = z.0 \land y = z.1) \lor \\
    & & (\exists z\ x = z.1 \land y = z.10)
  \end{eqnarray*}
  The infinite set
  $X = \{0, 1, 10, 11, 110, 111, \ldots\} = 1^*(0+1)$
  is not empty and for any $x$ in $X$ there is a $y$ in $X$
  such that $\psi(x,y)$ holds. Hence $\phi_r$ is true, 
  as should be shown by any decision procedure for S2S.  
  So $r$ loops.
  Note that $X$ is a closed recurrence set  for $r$. \qed
\end{example}

\begin{example}[S2S]
  Consider $r:~x \rightarrow \psi(x,y), y$ where
  \[\psi(x,y) =  (\forall X \ (x \in X \land \psi'(X)) \Rightarrow y \in X)\]
  with
  \[\psi'(X) = (\forall z\ z \in X \Rightarrow (z.0 \in X \land z.1 \in X))\]
  We have $\psi'(X)$ is true if and only if $X$ is closed by
  application of the successor functions $.0$ and $.1$. So,
  $\psi(x,y)$ is true if and only if $x\leq y$.
  MONA reports that $\phi_r$ is true, with a computed satisfying $X=\{\varepsilon\}$.
  By Corollary~\ref{dec-proc-w-computable-sufficient-conditions}, $r$ loops.
  \qed
\end{example}

\section{Related Works and Conclusion}
\label{sec:related-works-conclusion}

Recurrence sets were first introduced in~\cite{Gupta08} where $\psi$ denotes
any binary relation. Two symbolic analyses are presented in this paper
for constructing such sets: a bitwise analysis, which assumes that
the state space is finite and encoded using Boolean variables, and a
linear arithmetic analysis, which assumes that the program transitions
can be represented as rational linear constraints. In contrast to our
work, no second-order formulation is considered in this paper.

Let us now focus on termination-decidable classes of rules.
In \cite{DeSchreye89a}, the authors present a decision procedure for an 
arbitrary rule $\tilde{x} \rightarrow \psi(\tilde{x},\tilde{y}), \tilde{y}$
where $\psi(\tilde{x},\tilde{y})$ is a conjunction of equality constraints over rational trees.
In \cite{Lee01a, CodishLS05, BenAmram11a}, one finds variations of a decision procedure
for finite sets of rules $\tilde{x} \rightarrow \psi(\tilde{x},\tilde{y}), \tilde{y}$ 
where $\psi(\tilde{x},\tilde{y})$ is a conjunction of constraints $x > y$ or $x \geq y$ 
over a well-founded domain (such as the natural numbers) or the integers.
Generalizing \cite{Tiwari04}, termination of an arbitrary deterministic linear loop is
shown decidable  in \cite{Braverman06}  over the integers, the rationals, and the reals.
To the best of our knowledge, termination 
of a non-deterministic linear loop remains an open problem.

Summarizing the paper, we have seen that second-order logic is a convenient
language to express non-termination as a necessary and sufficient condition.
Such a condition is in general undecidable.
By restricting the language to the decidable cases S1S and S2S,
we have defined two new classes of rules, the termination of which is decidable.
Finally, we have shown that the weak versions of these logics provide sufficient conditions
for termination and non-termination of such rules.


\bibliographystyle{plain}

\end{document}